\documentclass[aps,twocolumn,showpacs,pra,superscriptaddress,floatfix]{revtex4}

\usepackage{graphicx}
\usepackage{amssymb}
\usepackage{amstext}
\usepackage{epsfig}

\usepackage{amsthm}
\usepackage{amsmath}
\usepackage{amsfonts}
\usepackage{amssymb,amstext}

\theoremstyle{plain}
\newtheorem{theorem}{Theorem}
\newtheorem{lemma}[theorem]{Lemma}

\theoremstyle{definition}

\usepackage{times,bbm,graphicx}
\usepackage{color}
\usepackage{tabularx}
\usepackage{epsfig}
\usepackage{mathtools}
\usepackage{amsmath}
\usepackage{amssymb}
\usepackage{graphicx}
\usepackage{wasysym}
\usepackage{dcolumn}
\usepackage{bm}
\usepackage{subfigure}
\usepackage{psfrag}
\usepackage{times}
\usepackage{bbm}

\begin{document}

\title{Pinning of Fermionic Occupation Numbers}

\author{Christian Schilling}
\affiliation{Institute for Theoretical Physics, ETH Z\"urich, Wolfgang--Pauli--Strasse 27, CH-8093 Z\"urich, Switzerland}

\author{David Gross}
\affiliation{Institute for Physics, University of Freiburg, Rheinstrasse 10, D-79104 Freiburg, Germany}

\author{Matthias Christandl}
\affiliation{Institute for Theoretical Physics, ETH Z\"urich, Wolfgang--Pauli--Strasse 27, CH-8093 Z\"urich, Switzerland}

\date{\today}

\begin{abstract}
	The Pauli exclusion principle is a constraint on the natural
	occupation numbers of fermionic states. It has been suspected
	since at least the 1970's, and only proved very recently, that there
	is a multitude of further constraints on these numbers, generalizing
    the Pauli principle.  Here, we provide the first
	analytic analysis of the physical relevance of these constraints. We
	compute the natural occupation numbers for the ground states of a
	family of interacting fermions in a harmonic potential.
	Intriguingly, we find that the occupation numbers are almost, but
	not exactly, pinned to the boundary of the allowed region
	(\emph{quasi-pinned}).  The result suggests that the physics behind
	the phenomenon is richer than previously appreciated.  In
	particular, it shows that for some models, the generalized Pauli
	constraints play a role for the ground state, even though they do
	not limit the ground-state energy.  Our findings suggest a
	generalization of the Hartree-Fock approximation.
\end{abstract}


\maketitle

\paragraph*{\label{sec:intro}Introduction.---}
In 1925 the study of atomic transitions led to Pauli's exclusion
principle \cite{Pauli1925}. It states that for identical fermions the
occupation number for any quantum state cannot exceed the value 1.  By
1926, Dirac \cite{Dirac1926} and Heisenberg \cite{Heis1926} had
identified the exclusion principle as a consequence of a much deeper
statement: the anti-symmetry of the many-fermion wave function.  While
anti-symmetry allows one to find the correct solutions to the full
many-fermion Schr\"odinger equation, it does not render the exclusion
principle obsolete: In many situations, the latter is sufficient to
predict the qualitative behavior of fermionic systems without the need
to resort to (often computationally intractable) \emph{ab initio}
methods.  The \emph{Aufbau} principle for elements in the periodic
table serves as a prime example.

This observation motivates the study of generalizations of the
exclusion principle, which, maybe surprisingly, exist and exhibit an
extremely rich structure \cite{Kly3}. To set the scene,
note that the Pauli constraint can be stated succinctly  as
\begin{equation}\label{PauliConstraint}
	0\leq \lambda_i \leq1\qquad
	\forall i
\end{equation}
in terms of the \emph{natural occupation numbers} $\lambda_i$, which
are the eigenvalues of the $1$-particle reduced density operator
($1-$RDO) $\rho_1$, normalized to the particle number $N$.
The utility of the exclusion principle is grounded in the fact that in
the ground states of many-fermion systems, one often observes $\lambda_i\approx 0$
or $\lambda_i \approx 1$, which is equivalent to stating that the
Hartree-Fock approximation works fairly well in these systems.

It had been observed in the 1970s that there are further linear
inequalities respected by the natural occupation numbers as a result of
global anti-symmetry \cite{Rus1,Rus2,Borl1972}.
One particular example is the so-called
Borland-Dennis setting $\wedge^3[\mathcal{H}_6]$ of three fermions and
a six dimensional $1-$particle Hilbert space $\mathcal{H}_6$
\cite{Borl1972}. Here, the set of constraints is given by
\begin{eqnarray}
	&&\lambda_1+\lambda_6 = \lambda_2+\lambda_5 = \lambda_3+\lambda_4 = 1 \label{d=6a} \qquad\\
	&&D^{(6)} := \lambda_5+\lambda_6-\lambda_4 \geq 0 \label{d=6b}
\end{eqnarray}
on the ordered eigenvalues $\lambda_i \geq \lambda_{i+1}$.

In a ground-breaking work building on recent progress in invariant
theory and representation theory, Klyachko exhibited an algorithm for
computing \emph{all} such Pauli-like constraints \cite{Kly3, Kly2}.
In fact, his work is part of a more general effort in quantum information
theory addressing the \emph{quantum marginal problem} which asks when a given set
of single-site reduced density operators (marginals) is compatible in the sense that they arise from
a common pure global state (see also \cite{Kly4,MC,Daftuar, MC2}). The global
state may be subject to certain symmetry constraints---one obtains the
fermionic case (commonly known as the $N-$representability problem \cite{Col, Col2}) by requiring total anti-symmetry under particle
exchange. Klyachko showed that for fixed particle number $N$
and dimension $d$ of the $1-$particle Hilbert space, the generalized
Pauli constraints amount to affine inequalities of the form
\begin{equation}\label{margconstraint}
\kappa_0 +\kappa_1 \lambda_1+\ldots+ \kappa_d \lambda_d \geq
0.
\end{equation}
Geometrically, these constraints define a \emph{convex polytope}
$\mathcal{P}_{N,d} \subset \mathbb{R}^d$ of possible spectra (for more
details see Appendix \ref{sec:notation}). In general, if a spectral inequality such
as (\ref{PauliConstraint}) or (\ref{margconstraint}) is (approximately)
saturated we say that the corresponding spectrum is (quasi-) pinned to
its extremum.

The natural question arises as to whether ground states of relevant
many-body models saturate some of those inequalities. Strong numerical
evidence supporting this conjecture has been presented in \cite{Kly1}.
The problem is challenging to address analytically, as one has to not
only compute the ground state, but also determine and diagonalize the
corresponding $1-$RDO.

Here, we present for the first time an analytic analysis. For the ground state
of a model of interacting fermions in a harmonic potential, the natural occupation
numbers are calculated. We obtain several results. We confirm that for this very natural
model, the natural occupation numbers lie indeed close to the boundary of
set of allowed ones. The analytic analysis enables us to track the
``trajectory'' of eigenvalues as a function of the interaction strength
between the fermions. What is conceptually also important, is
the fact that the eigenvalues never lie \emph{exactly on} the boundary. To
see why one could expect the opposite, note that the ground state
energy of a  Hamiltonian $H=\sum_{i,j} h^{(i,j)}$ with two-particle
terms $h^{(i,j)}$ can be represented as a constrained optimization
problem
\begin{equation*}
	E_{\min} =\min_{\rho_2^{(i,j)}}  \sum_{i,j=1}^N \operatorname{tr}[ h^{(i,j)} \rho_2^{(i,j)}]
\end{equation*}
where the $\rho_2^{(i,j)}$ are 2-particle density operators that are
\emph{compatible} in the sense that
they are the reduced densities of some $N-$fermion state \cite{Col, Col2}. Since the energy
functional is linear, it does not possess an unconstrained minimum. Therefore,
$E_{\min}$ will be achieved on the boundary of the set of compatible density
operators, where at least one of the compatibility constraints is
\emph{active} in the sense that any further minimization would violate
it. One way of understanding why a ``pinning'' effect for the natural
occupation numbers is observed, is to posit that the generalized Pauli
constraints are among the active physical constraints.
While this effect may well occur, we show
in this work that \emph{quasi}-pinning appears in natural fermionic
systems: the eigenvalue constraints seem to play a role, but are not
active in the above sense.  The finding suggests that the physics of
the phenomenon is richer than previously appreciated. We will return
to the physical consequences quasi-pinning has on the structure of wave
functions after presenting the calculations for our model system.

\paragraph*{The Model.---}
\label{sec:Model}
In order to analyze possible pinning effects
analytically, we consider a model of $N$ identical
fermions subject to a harmonic external potential and a harmonic
interaction term:
\begin{equation}\label{Hamiltonian}
H = \sum_{i=1}^N\,\left(\,\frac{p_i^2}{2m}+\frac{1}{2}m\omega^2 x_i^2\,\right) + \frac{1}{2} D\,\sum_{i,j=1}^N \, (x_i - x_j)^2 \,.
\end{equation}
The corresponding eigenvalue problem without any symmetry constraint
can easily be solved by transforming the Hamiltonian to the one of
decoupled harmonic oscillators. Two eigenfrequencies appear: a
non-degenerate one describing the center of mass motion and
another $(N-1)$-fold degenerate frequency associated with the relative
motion. The natural length scales corresponding to these eigenmodes are
\begin{equation}\label{length}
l := \sqrt{\frac{\hbar}{m
\omega}} , \qquad \tilde l := \sqrt{\frac{\hbar}{m \omega \sqrt{1+ N
D/(m\omega^2)}}}.\nonumber
\end{equation}
By rescaling the energy and the length scale, the
fermion-fermion coupling constant $D$ can be absorbed by the term $m
\omega^2$. Hence
the spectrum $\lambda$ of a $1-$RDO corresponding to
an eigenstate of $H$ depends only on the relative fermion-fermion
interaction strength $\frac{N\,D}{m \omega^2}= \left( \frac{l}{\tilde l}\right)^4-1.$
In fact, it will prove slightly more convenient to parameterize the coupling using
\begin{equation}\label{delta}
\delta:=\ln\left(\frac{l}{\tilde l}\right) = \frac{1}{4}\ln\left(1+\frac{N\,D}{m \omega^2}\right) .
\end{equation}
Then, in the regime of weak interaction, $D$ and $\delta$ are in leading order proportional,
$D=\frac{4m\omega^2}{N}\delta + O(\delta^2)$.

To study the physical relevance of the generalized Pauli constraints we restrict the Hamiltonian $H$ to the fermionic Hilbert space $\wedge^N[\mathcal{H}_{\infty}]$, with $\mathcal{H}_{\infty} = L^2(\mathbb{R})$, i.e. we are treating the $N$ particles as fermions (without spin). In \cite{harmOsc2012} $H$ has been diagonalized and the ground state reads in spatial representation ($\vec{x} = (x_1,\ldots,x_N)$)
\begin{eqnarray}\label{groundstate}
\lefteqn{\Psi_N(\vec{x}) = \mbox{const}\,\times\, \prod_{1\leq i<j\leq N}(x_i-x_j)} \\
&& \times \exp\left[-\frac{1}{2 N}\left(\frac{1}{{l}^2}-\frac{1}{{\tilde l}^2}\right) (x_1+\ldots+x_N)^2 -\frac{1}{2} \frac{1}{{\tilde l}^2} \vec{x}^2\right] \,. \nonumber
\end{eqnarray}
(Note its structural similarity to Laughlin's ground state wave
function describing the fractional quantum Hall effect
\cite{Laughlin1983}.
Moreover, the polynomial in front of the exponential function is the
Vandermonde determinant and by omitting it we obtain the ground
state in the bosonic $N$-particle Hilbert space.)

\paragraph*{The spectrum and its properties.---}
\label{sec:specctrum}
We now outline the calculation of the spectrum
$\lambda(\delta)$ as a function of the coupling.  We omit details of
this tedious but mostly straight-forward computation, presenting the
final result, together with some conceptual insights obtained along
the way.

The $1-$RDO is calculated by integrating out $N-1$
coordinates of the $N$-fermion state $\rho_N(\vec{x}, \vec{x}')
=\Psi_N^{\ast}(\vec{x}) \Psi_N(\vec{x}')$. An exercise in Gaussian
integration and integration by parts yields
\begin{equation}
\rho_1(x,x') = p(x,x')\, \exp[-\alpha (x^2+x'^2) + \beta x x']\,,\nonumber
\end{equation}
where $p$ is a symmetric polynomial of degree $\binom{N}{2}$ in the
variables $x,x'$ originating from the Vandermonde determinant in
(\ref{groundstate}), and $\alpha$ and $\beta$ some constants depending
on $l,\tilde l$ and $N$.

If the fermions do not interact with each other, the ground state
$|\Psi_N\rangle$ is a single Slater determinant and the spectrum of
its $1-$RDO is trivial, i.e.
\begin{equation}\label{slater}
	\lambda(\delta=0)=(\underbrace{1,\ldots,1}_N,0,\ldots).
\end{equation}
The regime of weak interaction can be characterized by the
condition $|D| \ll m \omega^2 $ or, equivalently,
$\delta \approx 0$.
We thus employ degenerate perturbation theory to obtain
$\lambda(\delta)$  around $\delta=0$.
The reason we employ the parameter $\delta$ is that one can prove a
duality
\begin{equation}\label{duality}
	\lambda_k\left(\delta\right)= \lambda_k\left(-\delta\right)\qquad\, \forall k\qquad
\end{equation}
relating the spectra for attractive ($\delta<0$) and repulsive
($\delta>0$) fermion-fermion interaction (interestingly, that this duality
holds is not obvious on the level of ground-state wave functions). This
immediately implies that the expansion $\lambda(\delta)$ contains only
even order terms, simplifying the  perturbation theory.

The solution for $N=3$ reads:
\small \begin{eqnarray}
1-\lambda_1 &= & \frac{40}{729} {\delta}^6 - \frac{1390}{59049} {\delta}^8 + O(\delta^{10}) \nonumber \\
1-\lambda_2 &= & \frac{2}{9} {\delta}^4 - \frac{232}{729}{\delta}^6 + \frac{3926}{10935} {\delta}^8 +O(\delta^{10}) \nonumber  \\
1-\lambda_3 &=& \frac{2}{9}{\delta}^4 - \frac{64}{243}{\delta}^6 + \frac{81902}{295245}{\delta}^8 +O(\delta^{10}) \nonumber \\
\lambda_4 &= & \frac{2}{9}{\delta}^4 - \frac{64}{243}{\delta}^6 + \frac{73802}{295245}{\delta}^8 + O(\delta^{10}) \nonumber \\
\lambda_5 &= &\frac{2}{9} {\delta}^4 - \frac{232}{729} {\delta}^6 + \frac{3976}{10935} {\delta}^8 +O(\delta^{10}) \nonumber \\
\lambda_6 &= & \frac{40}{729} {\delta}^6 - \frac{2200}{59049} {\delta}^8 + O(\delta^{10}) \nonumber \\
\lambda_7 &= & \frac{80}{2187} {\delta}^8 + O(\delta^{10}) \nonumber \\
\lambda_8 &=& O(\delta^{10}) \nonumber \\
\vdots \,& &\,\,\,\vdots \label{spectrum}
\end{eqnarray}
\normalsize
Similar results follow for $N=2$. Note the non-trivial \emph{hierarchy} of the eigenvalues,
\begin{equation}\label{hierachy}
\lambda_k = c_k \, \delta^{2 k - 6}  + O(\delta^{2 k - 4}) \qquad,
\end{equation}
for all $k\geq 5$.
Moreover, the spectrum $\lambda$ for $\delta$ not too large is very
close to the one of a single Slater determinant. For instance,
$\lambda_i$, $i=1,2,3$ deviate from $1$ and  $\lambda_j$, $j\geq 4$
from $0$ only by at most $1$ percent if $|\delta|\leq 0.5$. This
emphasizes the relevance of the Pauli constraints
(\ref{PauliConstraint}).

\paragraph*{Quasi-Pinning by Generalized Pauli Constraints.---}
\label{sec:pinning}
Equipped with the explicit solution (\ref{spectrum}), we can proceed
to analyze whether the generalized Pauli constraints play a role for the
ground state. While the underlying $1-$particle Hilbert space
$\mathcal{H}_{\infty}$ is infinite-dimensional, the scaling (\ref{hierachy})
implies that the spectrum is strongly
concentrated on a low-dimensional subspace, at least for small
$\delta$. One can use this fact to deduce statements about the
position of the total eigenvalues from truncated information alone.

This can be understood from simple geometric considerations.
Let $d<d'<\infty$. Because a $d$-dimensional $1-$particle
Hilbert space can be imbedded into any (larger) $d'$-dimensional one, one sees
that the convex polytope $\mathcal{P}_{N,d}$ is nothing but the intersection
between $\mathcal{P}_{N,d'}$ and the set of spectra with only $d$
non-zero eigenvalues (see also Appendix \ref{sec:truncation}).
Hence any facet of $\mathcal{P}_{N,d}$ arises
from the intersection of some facet of $\mathcal{P}_{N,d'}$
with the subspace of said spectra. 
Formally, a 
facet
$F'$ of
$\mathcal{P}_{N,d'}$ consists of points saturating a generalized Pauli 
constraint 
\begin{equation}
   D'(\lambda)=  \kappa_0 +
   \sum_{i=1}^{d} \kappa_i \lambda_i
   +
   \sum_{i=d+1}^{d'}\kappa_i \lambda_i \geq 0.
\end{equation}
Denote the first two summands by $D(\lambda^{\mathrm{tr}})$,
where $\lambda^{\mathrm{tr}}= (\lambda_i)_{i=1}^d$ is the truncated
spectrum.
Clearly, $D(\lambda^{\mathrm{tr}})=0$ describes the restriction of the
facet to the $d$-dimensional setting.
Now assume the truncated spectrum $\lambda^{\mathrm{tr}}(\delta)$ is not pinned, i.e.
$D(\lambda^{\mathrm{tr}}(\delta))>0$, then the hierarchical
scaling (\ref{hierachy}) implies	
\begin{equation}\label{truncated}
	D'\big(\lambda(\delta)\big)
	= D\big(\lambda^{\mathrm{tr}}(\delta)\big)+ O(\delta^{2d-4}),
\end{equation}
which is positive for $\delta$ small enough. Hence the full spectrum
$\lambda'$ also fails to be pinned. The case $d'=\infty$ works in the
same way, up to some mild assumptions (see Appendix \ref{sec:truncation}).

We will now apply these considerations to our model.  First, we
truncate to $6$ dimensions, which has the advantage that the spectral
polytope corresponding to $\wedge^3[\mathcal{H}_6]$ is 3-dimensional
and can thus be visualized. In a second step, we take a seventh
eigenvalue into account. This setting turns out to be strong enough
to establish all statements we have mentioned above -- namely
that the total spectrum is not exactly pinned, but does lie
close to the boundary (quasi-pinned).

\begin{figure}[!h]
\includegraphics[width=8cm]{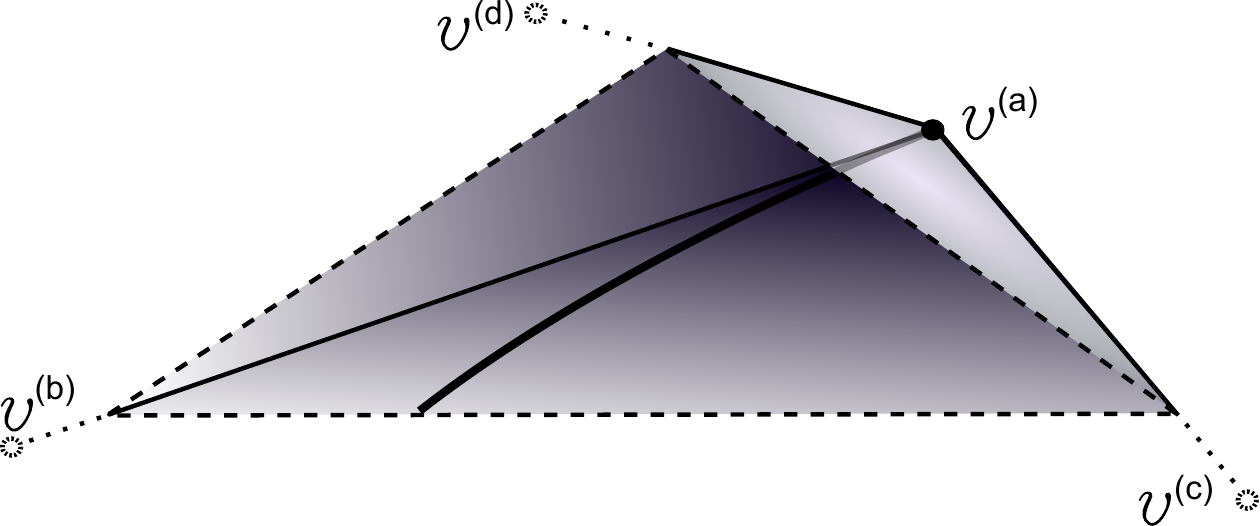}
\centering
\caption{Spectral ``trajectory'' $v(\delta)$ (thick line, partially
covered by facet, schematic) up to correction of order $\delta^8$ and
small part of the polytope $\mathcal{P}$ around vertex $v^{(a)}$
obtained by cutting $\mathcal{P}$ along the dashed lines}
\label{traj1}
\end{figure}

The simplest
non-trivial setting $\wedge^3[\mathcal{H}_6]$
becomes an appropriate description if
$\lambda_7,\lambda_8,\ldots \approx
0$. By (\ref{hierachy}), this condition is fulfilled if $\delta$ is small enough
that contributions of order $\delta^8$ can be neglected.
Choosing
$\lambda_4, \lambda_5$ and $\lambda_6$ as free parameters according to
(\ref{d=6a}), the corresponding polytope $\mathcal{P}_{3,6}$ effectively reduces
\cite{Kly1} to a 3-dimensional polytope $\mathcal{P}\subset \mathbb{R}^3$ with vertices,
\small \begin{eqnarray}
v^{(a)}= \,\,\left( 0,0,0\right)\,\,\,\,\,\,\,&,& v^{(b)}= \left( \frac{1}{2}\ \frac{1}{2},0\right) \nonumber \\
v^{(c)}= \left( \frac{1}{2} , \frac{1}{4} ,\frac{1}{4} \right)&,&v^{(d)}= \left( \frac{1}{2} , \frac{1}{2} ,\frac{1}{2}\right)\,.
\end{eqnarray}
\normalsize
Hence the vertex $v^{(a)}$ corresponds to single Slater determinants
and the 2-facet spanned by $v^{(a)},v^{(b)}$ and $v^{(c)}$ is defined
by $D^{(6)}=0$, which is here the one of interest and represents exact
pinning by constraint (\ref{d=6b}). We first illustrate schematically
our result (\ref{spectrum}) in Fig.~\ref{traj1}. There, the spectral ``trajectory'',
\begin{equation}
	v(\delta)=(\lambda_4(\delta),\lambda_5(\delta),\lambda_6(\delta)),\nonumber
\end{equation}
is shown as a thick line (neglecting effects of order $\delta^8$ and
higher).
It starts at the vertex
$v^{(a)}$ which corresponds  to the non-interacting situation $\delta=0$.
When increasing the fermion-fermion interaction, $v(\delta)$ leaves the
vertex $v^{(a)}$ and moves along the edge $(v^{(a)}, v^{(b)})$, the
distance to $v^{(a)}$ growing as $\delta^4$.
On the finer scale $\delta^6$,
$v(\delta)$ also moves away from the edge but is still \emph{pinned to the boundary}
of the polytope, lying on the 2-facet spanned by $v^{(a)}, v^{(b)}$ and
$v^{(c)}$. This is the bottom area in Fig.~\ref{traj1}, corresponding
to the constraint (\ref{d=6b}).

The pinning seems to disappear if we consider higher orders. From
(\ref{spectrum}), we can infer that the distance to the 2-facet $(v^{(a)}, v^{(b)},
v^{(c)})$ increases as $\delta^8$, \begin{equation}\label{pinning1}
D^{(6)}(\delta) = \zeta^{(6)} \,\delta^8 + O(\delta^{10})
\end{equation} with $\zeta^{(6)}= \frac{4510}{59049}$.
However, this calculation is inconclusive, as the distance to the
boundary is of the same order, $\delta^8$, as the truncation error (recall (\ref{truncated})).

To resolve the issue, we take another eigenvalue, $\lambda_7$, into
account. We thus work in the setting $\wedge^3[\mathcal{H}_7]$
with four constraints $D_i^{(7)}\geq 0$ for $i=1, \dots, 4$ \cite{Kly3}.
This setting is valid as long as $\lambda_8, \lambda_9,\ldots \approx0$ or
in other words we neglect terms of order $\delta^{10}$ or higher (but in
contrast to the setting $\wedge^3[\mathcal{H}_6]$ we include
$\delta^8-$terms). Since the polytope is now $6-$dimensional we cannot
present our results graphically anymore. The results (\ref{spectrum})
lead to ($i=1,2,3,4$)
\begin{equation}
D_i^{(7)}  = \zeta_i^{(7)}{\delta}^8+ O({\delta}^{10}) \,\,\,,\label{pinning2}
\end{equation}
with $\zeta_1^{(7)}=\frac{20 }{2187}$, $\zeta_2^{(7)}=\frac{10 }{243}$,
$\zeta_3^{(7)}=\frac{50 }{2187}$, $\zeta_4^{(7)}=\frac{2890 }{59049}$.
Here in the $\wedge^3[\mathcal{H}_7]$-analysis, the new result is that
all four distances $D^{(7)}_i$ are non-zero to a smaller order, $\delta^8$,
than the error of spectral truncation, $\delta^{10}$.
Together with the comments at the beginning of this section, this
shows that the absence of pinned spectra is genuine, rather than an artifact of
the truncation.
Given this, the quasi-pinning found here is surprisingly strong.  In
particular it exceeds by four additional orders the
(quasi-)pinning by Pauli's exclusion principle constraints
(\ref{PauliConstraint}),
\begin{equation}\label{PauliConstraint2}
0\leq 1-\lambda_2(\delta), 1-\lambda_3(\delta), \lambda_4(\delta),
\lambda_5(\delta) = \frac{2}{9}\,\delta^4 + O(\delta^{6}).
\end{equation}

\paragraph*{Generalizing Hartree-Fock.---}
In this section, we discuss what conclusions can be
drawn about the $N-$fermion state $|\Psi\rangle$ itself, given information just
about the position of the eigenvalues of the corresponding 1-RDO relative to the boundary
of the spectral polytope. In this way, quasi-pinned spectra are endowed
with a physical significance.
To this end, recall the basic fact that the spectrum
$\lambda_{\mathrm{Sl}}=(1, \dots, 1, 0, \dots ,0)$ can \emph{only}
arise from a Slater determinant $|\Psi\rangle=|1,\ldots,N\rangle$.
It is well-known that this statement is stable under small
deviations: if $\lambda\approx \lambda_{\mathrm{Sl}}$, then $|\Psi\rangle$ is
well-approximated by a Slater determinant (see \cite{Bach1992} or Appendix \ref{sec:quasiSelectionRule}).

For exactly pinned spectra, there is a simple generalization of these
observations. In \cite{Kly1}, it is stated that constraint
(\ref{d=6b}) can be saturated only by states of the form
\begin{equation*}
	|\Psi\rangle = \alpha |1,2,3\rangle + \beta |1,4,5\rangle + \gamma
	|2,4,6\rangle,\nonumber
\end{equation*}
a fact is dubbed ``selection rule'' for Slater determinants (see also
Appendix \ref{sec:selectionrule}). The general statement reads: if
$D(\lambda)\geq 0$ is a generalized Pauli constraint, then $D(\lambda)=0$ can
only be achieved by states $|\Psi\rangle$ which are superpositions of those
Slater determinants whose (unordered) spectra also saturate $D$.

What is more important, a stable version of this statement applying to
quasi-pinned states can be found---at least for specific situations. In
the Appendix \ref{sec:quasiSelectionRule}, we show that for the Borland-Dennis
setting, spectra in the vicinity of the facet corresponding to constraint (\ref{d=6b})
are approximately of the form above. In particular, quasi-pinned states are close
to states containing fairly low amounts of multi-partite entanglement as
quantified by the Schmidt number \cite{Eisert2001}. In \cite{Vrana2008} a new entanglement
measure has been suggested, which, for the Borland-Dennis setting, naturally separates exactly
pinned and non-pinned states. We believe that these
findings open up a potentially significant avenue for investigating the
structure of fermionic ground states via their natural occupation numbers---
generalizing a program that has long been carried out for the Hartree-Fock case
\cite{Bach1992}.

We close by speculating that these insights could give rise to
improved numerical procedures.
The idea is to replace the
ground state ansatz of one single Slater determinant by the states
corresponding to the points lying on the (quasi-)pinning polytope
facet. In contrast to the configuration interaction (CI) methods in
quantum chemistry which improve the Hartree-Fock approximation by
adding several arbitrary Slater determinants to the Hartree-Fock state
our method would add only a few but carefully chosen additional Slater
determinants.

\paragraph*{Conclusions.---} \label{sec:conc} For a natural model of
interacting fermions in a harmonic trap we analytically calculated the
leading orders of the eigenvalues of the $1-$RDO corresponding to the
fermionic ground state as function of $\delta$, a measure for the
fermion-fermion interaction strength.  The investigation of the
generalized Pauli constraints has shown that none of them is completely
saturated, which might be a generic property of all continuous models of
interacting fermions. In particular, the findings show that it is likely extremely
challenging to use numerical methods to distinguish between genuinely
pinned and mere quasi-pinned states. This underscores the need for
analytical analyses, first provided here.
On the other hand the pinning up to corrections of order $\delta^8$ we
found here is surprisingly strong. In particular it exceed the one by
the Pauli exclusion principle constraints (\ref{PauliConstraint}), which are pinning
up to corrections of order $\delta^4$ only.

\paragraph*{Acknowledgements.---}
We thank F.\hspace{0.5mm}Verstraete for helpful discussions.
CS and MC acknowledge support from the Swiss National Science
Foundation (grants PP00P2-128455 and 20CH21-
138799), the National Centre of Competence in Research
`Quantum Science and Technology' and the German Science
Foundation (grant CH 843/2-1). DG's research is supported by
the Excellence Initiative of the German Federal and State
Governments (grant ZUK 43).

\section*{Appendix}
This appendix is split into four sections. The first one introduces the notation and repeats the solution of the fermionic quantum marginal problem. In the second section we explain how to simplify the pinning analysis by truncating the spectrum. This amounts to the proof of statement (13), a relation connecting polytope distances of the correct and truncated marginal setting. The third section introduces a selection rule, which explains how the structure of a $N-$fermion state simplifies if its natural occupation numbers are exactly pinned to some Pauli facet and applies it to the Borland-Dennis setting. In the last section we present a modification of this selection rule for the case of only approximate pinning. This then justifies our Hartree-Fock generalization.

\subsection{Notation and Fermionic Quantum Marginal Problem.---}
\label{sec:notation}
The problem of determining all spectra
\begin{equation}\label{specordered}
\lambda=(\lambda_i)_{i=1}^{d'}\qquad, \,1\geq \lambda_1\geq \lambda_2\geq \ldots \geq \lambda_{d'}\geq 0  \qquad
\end{equation}
of $1-$particle reduced density operators ($1-$RDO) $\rho_1$ arising from some pure $N-$fermion state $|\Psi\rangle \in \wedge^N[\mathcal{H}_{d'}]$,
\begin{equation}
\rho_1 \equiv N\,\mbox{tr}_{N-1}[|\Psi\rangle \langle \Psi|]
\end{equation}
by tracing out $N-1$ particles, is known as the fermionic quantum marginal problem of the setting $\wedge^N[\mathcal{H}_{d'}]$. Here $d' \in \mathbb{N}\cup \{\infty\}$, $\mathcal{H}_{d'}$ is the $d'-$dimensional separable $1-$particle Hilbert space and we use the trace normalization convention,
\begin{equation}\label{norm}
\mbox{tr}[\rho_1] = \lambda_1+\ldots+\lambda_{d'} = N ,
\end{equation}
common in quantum chemistry.

For $d'$ finite, the family of possible spectra (we call them compatible w.r.t $\wedge^N[\mathcal{H}_{d'}]$),
is described by finitely many independent conditions $\{C_i\}$, the generalized Pauli constraints. Each of them has the form
\begin{equation}\label{margcon}
C_i\,:\,\,\,D_i(\lambda) = \kappa_0 + \kappa_1 \lambda_1+\ldots \kappa_{d'} \lambda_{d'} \geq 0 ,
\end{equation}
$\kappa_0,\ldots,\kappa_{d'} \in \mathbb{Z}$ and describes a half-space $V_i$ of $\mathbb{R}^{d'}$. These constraints together with the trivial conditions (\ref{specordered}) and (\ref{norm}) define the polytope $\mathcal{P}_{N,d'} \subset \mathbb{R}^{d'}$
of possible spectra. In that sense every constraint (\ref{margcon}) gives rise to a facet $F_i$ of this polytope,
\begin{equation}
F_i = \{\lambda \in \mathcal{P}_{N,d'}\,|\,D_i(\lambda)= 0\}.
\end{equation}
Note that besides these Pauli facets there are also further facets, those corresponding to the trivial constraints (\ref{specordered}), but they will not be of interest in our work. Moreover, the quantity $D_i(\cdot)$, which is only defined up to a positive factors, defines after fixing this factor (i.e. the parameters $\kappa_i$) a measure for the distance of spectra to the corresponding facet $F_i$. For the case of $d'$ finite it coincides up to normalization with the Euclidean distance, $\mbox{dist}_2(\mu,F_i)=\frac{D_i(\mu)}{\|\kappa\|_2}$, $\kappa= (\kappa_1,\ldots,\kappa_{d'})$.

For the case $d'=\infty$ the set $\mathcal{P}_{N,\infty}$ of compatible spectra is not explicitly known yet. Nevertheless, for our work we assume  that it is also defined by a family of linear inequalities
\begin{equation}
D_j^{(\infty)}(\lambda) = \kappa_0+\kappa_1 \lambda_1+ \kappa_2 \lambda_2+ \ldots \geq 0\,.
\end{equation}
The results on truncation of the spectrum and the relation of polytope $\mathcal{P}_{N,d}$ and $\mathcal{P}_{N,d'}$, $d<d'$ finite presented in Appendix \ref{sec:truncation} strongly emphasizes that this assumption is justified. Moreover, the involved fact that the $l^1-$closure $\overline{\mathcal{P}}_{N,d}$ is convex also suggests this assumption.

Finally, we still make some comments on the meaning of natural orbitals $\{|k\rangle\}$, the eigenvectors of the $1-$RDO,
\begin{equation}
\rho_1 = \sum_{k=1}^{d'} \, \lambda_k \, |k\rangle \langle k| ,
\end{equation}
and their utility for applications.

These natural orbitals induced by a fixed state $|\Psi\rangle \in \wedge^N[\mathcal{H}_{d'}]$, $d' \in \mathbb{N}\cup \{\infty\}$, define a basis $\mathcal{B}_1:=\{|k\rangle\}_{k=1}^{d'}$ for the $1-$particle Hilbert space $\mathcal{H}_{d'}$. For ease of notation we skip the argument $\Psi$ of $|i(\Psi)\rangle$. Basis $\mathcal{B}_1$ then induces the basis $\mathcal{B}_N$ for $\wedge^{N}[\mathcal{H}_{d'}]$ of corresponding Slater determinants ($1\leq i_1<\ldots<i_N\leq d'$)
\begin{equation}
|\mathbf{i}\rangle \equiv |i_1,\ldots,i_N\rangle \equiv \mathcal{A}_N [|i_1\rangle \otimes \ldots \otimes |i_N\rangle] ,
\end{equation}
where $\mathcal{A}_N$ is the anti-symmetrizing operator on the $N-$particle Hilbert space $\mathcal{H}_{d'}^{\,\,\,\otimes^N}$.
By expanding $|\Psi\rangle$ w.r.t. to $\mathcal{B}_N$,
\begin{equation}\label{expansion}
|\Psi\rangle = \sum_{\mathbf{i}}\,c_{\mathbf{i}}\,|\mathbf{i}\rangle
\end{equation}
the natural occupation numbers are given by
\begin{equation}\label{noncoef}
\lambda_k = \sum_{\mathbf{i},\,k \in \mathbf{i}}\,|c_{\mathbf{i}}|^2 .
\end{equation}
To compare marginal settings of different dimensions, $d, d'$ with $d < d'\leq \infty$ we imbed $\mathcal{H}_{d}$ into $\mathcal{H}_{d'}$,
\begin{equation}
\mbox{span}\{|i\rangle\}_{i=1}^{d}\equiv\mathcal{H}_{d} \leq \mathcal{H}_{d'} \equiv \overline{\mbox{span}\{|i\rangle\}_{i=1}^{d'}} ,
\end{equation}
where the closure is only relevant for the case $d'$ infinite.
In the same way,
\begin{equation}
\wedge^N[\mathcal{H}_{d}] \leq \wedge^N[\mathcal{H}_{d'}] .
\end{equation}
Indeed, according to (\ref{expansion}), we find that every state
\begin{equation}\label{expansionsmall}
|\Psi\rangle =  \sum_{1\leq i_1 <\ldots<i_N \leq d} c_{\mathbf{i}} \,|\mathbf{i}\rangle \,\,\,\,\in \wedge^N[\mathcal{H}_{d}]
\end{equation}
can be imbedded into $\wedge^N[\mathcal{H}_{d'}]$ by
\begin{equation}\label{expansionlarge}
|\Psi'\rangle =  \sum_{1\leq i_1 <\ldots<i_N \leq d} c_{\mathbf{i}} \,|\mathbf{i}\rangle \,\,\,\,\in \wedge^N[\mathcal{H}_{d'}]\,,
\end{equation}
and all the other coefficients $c_{\mathbf{i}}$ in (\ref{expansionlarge}), those with $i_N>d$, vanish. We used here different symbols for the states $|\Psi\rangle$ and $|\Psi'\rangle$ to distinguish between the two different spaces $\wedge^N[\mathcal{H}_{d}]$ and $\wedge^N[\mathcal{H}_{d'}]$ to which they belong. This subtle difference is becoming relevant if we determine the natural occupation numbers $\lambda'$ of $|\Psi'\rangle$ (recall (\ref{noncoef})), \begin{equation}\label{nonimbedded}
\lambda' = (\lambda_1,\ldots,\lambda_d,\underbrace{0,\ldots,0}_{d'-d})
\end{equation}
differing from $\lambda=(\lambda_1,\ldots\lambda_d)$ by additional zeros. In the following, to simplify the notation, we will use the same symbols for mathematical objects and their imbeddings into larger spaces.

\subsection{Truncation of the Spectrum.---}\label{sec:truncation}
In our work we have determined the ``trajectory'' of spectra
\begin{equation}
\lambda(\delta) = (\lambda_i(\delta))_{i=1}^{\infty} \in \mathcal{P}_{3,\infty}  ,
\end{equation}
of the $1-$RDO corresponding to the ground state of a $3-$fermion model with relative fermion-fermion interaction strength $\delta$. The goal was then to show that for $\delta$ not too large, $\lambda(\delta)$ is almost but not exactly saturating some of the generalized Pauli constraints of its setting $\wedge^N[\mathcal{H}_{\infty}]$. Geometrically this means that the vector $\lambda(\delta)$ is very close to some Pauli facet $F_i$ of $\mathcal{P}_{3,\infty}$. In that case we say that the spectrum is quasi-pinned to the facet $F_i$.
Since $\mathcal{P}_{3,\infty}$ is not explicitly known and quite involved (it is described by infinitely many constraints on infinitely many eigenvalues), we have truncated the spectrum and simplified the pinning analysis by considering only the largest $d$ eigenvalues,
\begin{equation}\label{spectrunc}
\lambda^{\mathrm{tr}} = (\lambda_1,\ldots,\lambda_d) ,
\end{equation}
and analyzed the saturation of the constraints corresponding to the setting $\wedge^3[\mathcal{H}_{d}]$. The following fact justifies this approach:  For $d<d'$ every Pauli facet $F$ of $\mathcal{P}_{N,d}$ is contained in some Pauli facet $F'$ of $\mathcal{P}_{N,d'}$, i.e. $F$ is the intersection of $F'$ with the hyperplane of spectra with only $d$ non-zero eigenvalues.
Then, for small $\lambda_{d+1},\lambda_{d+2},\ldots$, small distance of $\lambda^{\mathrm{tr}}$ to $F$ translates to small distances of $\lambda$ to $F'$ modulo an error of order of the largest neglected eigenvalue, $\lambda_{d+1}$.
To illustrate this, we present the example $\wedge^3[\mathcal{H}_6]$, which is one of the two settings studied in our work. There one generalized Pauli constraint reads \cite{Borl1972, Rus1, Rus2, Kly3}
\begin{equation}
D^{(6)}(\lambda):=2-(\lambda_1+\lambda_2+\lambda_4) \geq 0 \label{d=6g}\,\,.
\end{equation}
For the setting $\wedge^3[\mathcal{H}_{\infty}]$ the known constraint \cite{Kly3}
\begin{equation}\label{margconinf}
D^{(\infty)}(\lambda)=2-(\lambda_1+\lambda_2+\lambda_4+\lambda_7+\lambda_{11}+\lambda_{16} +\ldots )\geq 0\,,
\end{equation}
coincides with constraint (\ref{d=6g}) up to a linear combination of eigenvalues $\lambda_7,\lambda_{11},\lambda_{16},\ldots$, which where neglected in the truncated setting.

A first important step in proving the universality of this relation between polytope distances of correct and truncated setting is the next lemma:
\begin{lemma}\label{lemzeros}
Consider the quantum marginal problems of the two settings $\wedge^N[\mathcal{H}_d]$ and $\wedge^N[\mathcal{H}_{d'}]$, $d < d' \in \mathbb{N}\cup \{\infty\}$ and let $\lambda = (\lambda_1,\ldots,\lambda_d)$ be a spectrum. Then,
\begin{eqnarray}
(\lambda_1,\ldots,\lambda_d)\,\, \mbox{compatible w.r.t.} \,\,\wedge^N[\mathcal{H}_d] &&\nonumber \\
\Leftrightarrow \qquad \qquad\qquad \nonumber &&\\
(\lambda_1,\ldots,\lambda_d,\underbrace{0,\ldots,0}_{d'-d})\,\, \mbox{compatible w.r.t.}\,\, \wedge^N[\mathcal{H}_{d'}] &&\,\,.
\end{eqnarray}
For the corresponding polytopes this means
\begin{equation}
\mathcal{P}_{N,d} = \mathcal{P}_{N,d'}|_{\lambda_{d+1},\lambda_{d+2},\ldots=0} ,
\end{equation}
the polytope $\mathcal{P}_{N,d'}$ intersected with the hyperplane
given by $\lambda_{d+1},\lambda_{d+2},\ldots=0$ coincides with
$\mathcal{P}_{N,d}$.
\end{lemma}
\begin{proof}
The direction ``$\Rightarrow$'' was already explained at the end of Section\ref{sec:notation}.
To prove ``$\Leftarrow$'' we show that a state $|\Psi'\rangle$ expanded according to (\ref{expansion}),
\begin{equation}
|\Psi'\rangle =  \sum_{1\leq i_1 <\ldots<i_N \leq d'} c_{\mathbf{i}} \,|\mathbf{i}\rangle \,\,\,\,\in \wedge^N[\mathcal{H}_{d'}]\,,
\end{equation}
with natural occupation numbers $(\lambda_1,\ldots,\lambda_d,\underbrace{0,\ldots,0}_s)$ contains only Slater determinants $|\mathbf{i}\rangle$, with $i_1,\ldots,i_N\leq d$. But this is clear
due to (\ref{noncoef}), which then yields
\begin{equation}
\forall \,k\,>d\,:\,\,\, 0 \stackrel{!}{=}\lambda_k = \sum_{\mathbf{i},\,k \in \mathbf{i}}\,|c_{\mathbf{i}}|^2 .
\end{equation}
Hence $c_{\mathbf{i}} = 0$ if $i_N>d$.
\end{proof}
What does Lemma \ref{lemzeros} imply for the relation between the families of generalized Pauli constraints of two settings?
Let us consider two settings with $d,d'$ finite, $d<d'$. Every constraint $D_j'$ for the setting $\wedge^N[\mathcal{H}_{d'}]$
is linear and hence its restriction
\begin{equation}\label{constrestricted}
\hat{D}_j'(\lambda_1,\ldots,\lambda_d) \equiv D_j'(\lambda_1,\ldots,\lambda_d,0,\ldots)  \geq 0
\end{equation}
to the hyperplane defined by $0=\lambda_{d+1},\lambda_{d+2},\ldots$ is also a linear constraint in the remaining coordinates $\lambda_1,\ldots,\lambda_d$. How is the half space $V_j\subset \mathbb{R}^d$ corresponding to (\ref{constrestricted}) related to the polytope $\mathcal{P}_{N,d}$? Lemma \ref{lemzeros} states that
\begin{equation}
\mathcal{P}_{N,d} \subset V_j
\end{equation}
and
\begin{equation}
\mathcal{P}_{N,d} = \cap_j V_j|_{\ast} ,
\end{equation}
where the star $\ast$ denotes here the restriction to spectra, i.e. ordered and normalized vectors.
There are two possible relations between $V_j$ (or $V_j|_{\ast}$) and $\mathcal{P}_{N,d}$. They are illustrated in Figure \ref{facetproj} in form of a simplified $2-$dimensional picture:
\begin{figure}[!h]
\includegraphics[width=8cm]{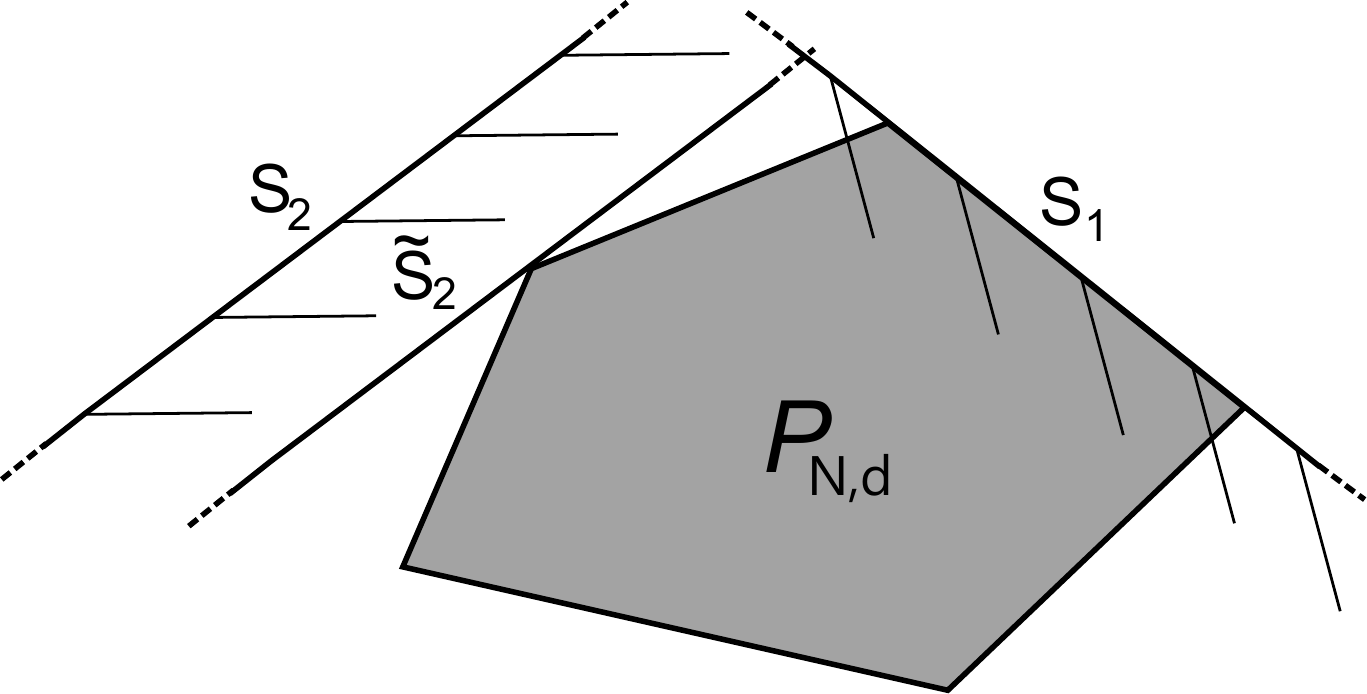}
\centering
\caption{Polytope $\mathcal{P}_{N,d}$ and two restricted generalized Pauli constraints $\hat{D}_1, \hat{D}_2 \geq 0$ with boundaries $S_1, S_2$ arising from two constraints $D_1,D_2\geq 0$ belonging to a higher dimensional marginal settings $\wedge^N[\mathcal{H}_{d'}]$.}
\label{facetproj}
\end{figure}
There, we consider two half spaces $V_1$ and $V_2$ corresponding to the ``restricted'' constraints $\hat{D}'_1\geq 0$ and $\hat{D}'_2\geq 0$ with boundaries $S_1$ and $S_2$ and orientation indicated by stripes. Such hyperplanes can either contain a facet of maximal (example $S_1$) or lower dimension of $\mathcal{P}_{N,d}$ or they lie outside of $\mathcal{P}_{N,d}$ (example $S_2$). The third case of a proper intersection is not possible due to Lemma \ref{lemzeros}. Every constraint $D'$ with boundary $S$ of its restriction $\hat{D}'$ lying outside of $\mathcal{P}_{N,d}$ is a constraint, which is irrelevant
for the pinning analysis since it has the form
\begin{equation}
D'(\lambda) = c + \tilde{D}(\lambda^{\mathrm{tr}}) + O(\lambda_{d+1})\,,
\end{equation}
where $\tilde{D}(\lambda^{\mathrm{tr}})\geq 0$ is a constraint of the setting $\wedge^N[\mathcal{H}_d]$ with a boundary shown in Figure \ref{facetproj} as hyperplane $\tilde{S}_2$ and $c >0$ is some offset.
Hence if the spectrum decays sufficiently fast, constraint $D'$ is not saturated at all due to the offset $c$ and thus irrelevant.
Moreover, for every Pauli facet of $\mathcal{P}_{N,d}$ corresponding to some constraint $D>0$, Lemma \ref{lemzeros} guarantees the existence of a constraint $D'>0$ in the larger setting whose projection $\hat{D'}$ coincides with $D$.
We summarize these insights by stating
\begin{lemma}\label{lemdistancemodif}
Given two marginal settings $\wedge^N[\mathcal{H}_d]$ and $\wedge^N[\mathcal{H}_{d'}]$ with $d < d' \in \mathbb{N}$. Every generalized Pauli constraint $D'\geq 0$ of the setting $\wedge^N[\mathcal{H}_{d'}]$ relevant for the pinning analysis is given by a linear modification of some generalized Pauli constraint $D\geq 0$ of the setting $\wedge^N[\mathcal{H}_{d}]$,
\begin{equation}\label{distancemodif}
D'(\lambda) = D(\lambda^{\mathrm{tr}}) + O(\lambda_{d+1}) .
\end{equation}
\end{lemma}
Finally, we remark that for the important case $d'$ infinite effectively the same results holds but one has to deal with one subtlety. Since $\mathcal{P}_{N,\infty}$ is described by infinitely many constraints Lemma \ref{lemzeros} guarantees for every constraint $D\geq 0$ of the setting $\wedge^N[\mathcal{H}_d]$, only the existence of a sequence of constraints $D'_j\geq 0$ whose restrictions $\hat{D}'_j\geq 0$ converge to the constraint $D\geq 0$. This means that condition (\ref{distancemodif}) in Lemma \ref{lemdistancemodif} holds up to a small error $\varepsilon$,
\begin{equation}
D_{\varepsilon}'(\lambda) = \varepsilon + D(\lambda^{\mathrm{tr}}) + O(\lambda_{d+1}) ,
\end{equation}
which can be made arbitrarily small by choosing appropriate constraints $D'_{\varepsilon}$. Hence, to minimize the technical effort we assume in our work that Lemma \ref{lemdistancemodif} holds in its original form also for the case $d'$ infinite.

\subsection{\label{sec:selectionrule} Selection Rule.---}
In this section we state a selection rule which explains how the structure of the $N-$fermion state
$|\Psi\rangle \in \wedge^N[\mathcal{H}_d]$ simplifies if the spectrum of the corresponding $1-$RDO is pinned to
some Pauli facet of $\mathcal{P}_{N,d}$. Moreover, we apply it for the setting $\wedge^3[\mathcal{H}_6]$.

Let's consider a state $|\Psi\rangle$ with natural occupation numbers $\lambda = (\lambda_i)_{i=1}^d$ saturating some generalized Pauli constraint
\begin{equation}\label{margcon2}
D(\lambda) = \kappa_0 + \kappa_1 \lambda_1+\ldots \kappa_{d} \lambda_{d} \geq 0 .
\end{equation}
In \cite{Kly1}, by introducing the creation and annihilation operator $a^{\dagger}_k, a_k$ of a fermion in the natural orbital $|k\rangle$ and the particle number operators $N_k \equiv a^{\dagger}_k a_k$,
an important condition is stated, which $|\Psi\rangle$ in that case satisfies:
\begin{equation}\label{pinningeigen}
\hat{D}|\Psi\rangle \equiv \left(\kappa_0 \mathrm{Id} + \kappa_1 N_1+\ldots \kappa_{d} N_{d}\right) |\Psi\rangle = 0 .
\end{equation}
Applying this condition to the expansion of $|\Psi\rangle$ in Slater determinants induced by the natural orbitals,
\begin{equation}\label{Psiansatz}
|\Psi\rangle =  \sum_{\mathbf{i}} c_{\mathbf{i}} \,|\mathbf{i}\rangle \qquad
\end{equation}
it implies \emph{Klyachko's selection rule}, which states that whenever
\begin{equation}
\hat{D}|\mathbf{i}\rangle \neq 0 ,
\end{equation}
the corresponding coefficient $c_{\mathbf{i}}$ vanishes.
To show the strength of this selection rule we study states in the Borland-Dennis setting. The corresponding Hilbert space
$\wedge^3[\mathcal{H}_6]$ has dimension $\binom{6}{3}=20$ and the generalized Pauli constraints read \cite{Borl1972, Rus1, Rus2, Kly3}
\begin{eqnarray}
	&&\lambda_1+\lambda_6, \,\lambda_2+\lambda_5, \,\lambda_3+\lambda_4 \leq 1 \label{d=6c} \qquad\\
	&&D^{(6)} := 2-(\lambda_1 +\lambda_2+\lambda_4) \geq 0 \label{d=6d} .
\end{eqnarray}
The normalization together with the non-negativity of the eigenvalues leads to
\begin{equation}
\lambda_1+\lambda_6 =\lambda_2+\lambda_5 = \lambda_3+\lambda_4 = 1 \label{d=6e} .
\end{equation}
Hence the constraints in (\ref{d=6c}) are always saturated and this implies according to (\ref{pinningeigen})
\begin{eqnarray}\label{d=6f}
\left(\mathrm{Id}-N_1-N_6\right)|\Psi\rangle &=& 0 \nonumber \\
\left(\mathrm{Id}-N_2-N_5\right)|\Psi\rangle &=& 0 \nonumber \\
\left(\mathrm{Id}-N_3-N_4\right)|\Psi\rangle &=& 0 .
\end{eqnarray}
Klyachko's selection rule applied to (\ref{d=6f}) implies that every Slater determinant
showing up in the ansatz (\ref{Psiansatz}) for $|\Psi\rangle$ is built up by natural orbitals with one index from each set $\{1,6\}, \{2,5\}$ and $\{3,4\}$.
Those are the $8$ states $|1,2,3\rangle$, $|1,2,4\rangle$, $|1,3,5\rangle$, $|1,4,5\rangle$, $|2,3,6\rangle$,
$|2,4,6\rangle$, $|3,5,6\rangle$ and $|4,5,6\rangle$. If the constraint (\ref{d=6d}) is also saturated the selection rule
restricts this family of Slater determinants to the three states $|1,2,3\rangle$, $|1,4,5\rangle$ and $|2,4,6\rangle$ and in
that case we find
\begin{equation}\label{HFextstate}
|\Psi_3\rangle = \alpha |1,2,3\rangle + \beta |1,4,5\rangle + \gamma |2,4,6\rangle .
\end{equation}

\subsection{Quasi-Pinning and modified Selection Rule.---}\label{sec:quasiSelectionRule}
In this section we show for the Borland-Dennis setting that any state $|\Psi\rangle \in \wedge^3[\mathcal{H}_6]$ whose natural
occupation numbers are approximately saturating the corresponding generalized Pauli constraint (\ref{d=6d}) also fulfill approximately
condition (\ref{pinningeigen}). We also quantify this relation. This result then guarantees that our Hartree-Fock extension will work
for systems exposing strong pinning.

As a warm-up and since we will need the result we first study a simpler question.
It is a basic fact that the spectrum $\lambda_{\mathrm{Sl}}=(1, \dots, 1, 0, \dots ,0)$ can \emph{only}
arise from a Slater determinant $|\Psi\rangle=|1,\ldots,N\rangle$. Is this statement stable under small
deviations, i.e. $\lambda\approx \lambda_{\mathrm{Sl}} \Rightarrow|\Psi\rangle \approx |1,\ldots,N\rangle$? Yes, it is true according to
\begin{lemma}\label{Slaterstable}
Consider a state $|\Psi\rangle \in \wedge^N[\mathcal{H}_d]$, let $\{|k\rangle\}_{k=1}^d$ be its natural orbitals
and denote the projection operator onto the space spanned by $|1,\ldots,N\rangle$ by $P_{\mathrm{Sl}}$.
Then,
\begin{equation}
1-\delta \leq \|P_{\mathrm{Sl}} \Psi  \|_{L^2}^2 \leq 1- \frac{1}{N} \delta ,
\end{equation}
where
\begin{equation}
0\leq N-(\lambda_1+\ldots+\lambda_N) =:\delta .
\end{equation}
\end{lemma}

\begin{proof}
We expand the state $|\Psi\rangle$ in Slater determinants induced by natural orbitals (recall Section \ref{sec:notation}),
\begin{equation}\label{expansionstate1}
|\Psi\rangle = \sum_{\textbf{i}}\,c_{\textbf{i}}\,|\textbf{i}\rangle .
\end{equation}
We define the operator
\begin{equation}
\hat{S}  = N\, \mathrm{Id}-\left(a_1^{\dagger} a_1+\ldots + a_N^{\dagger} a_N\right) .
\end{equation}
Since all operators $a_i^{\dagger} a_i, i=1,\ldots, d$ commute it is clear that $\hat{S}$ has the spectrum $\{0,1,\ldots,N\}$ with eigenstates $|\textbf{i}\rangle$.
The eigenvalue corresponding to $|\textbf{i}\rangle$ is the number of indices $k \in \textbf{i}$ not belonging to the set $\{1,\ldots,N\}$. We denote the set of indices leading to the eigenvalue $k$ by $J_k$ and find
\begin{eqnarray}
\delta &\equiv& N- (\lambda_1+\ldots+\lambda_d ) \nonumber \\
&=& \langle \Psi|N \, \mathrm{Id} -\left( a_1^{\dagger}a_1+\ldots+ a_d^{\dagger}a_d\right) \,|\Psi\rangle \nonumber \\
&\equiv& \langle \Psi|\hat{S}|\Psi\rangle \nonumber \\
&=& \sum_{\textbf{i},\textbf{j}\in J_0\cup \ldots \cup J_N } c_{\textbf{j}}^{\ast}\,c_{\textbf{i}}\,\langle \textbf{j} |\hat{S} |\textbf{i}\rangle \nonumber \\
&=& \sum_{\textbf{i}\in J_0\cup \ldots \cup J_N } |c_{\textbf{i}}|^2\,\langle \textbf{i} |\hat{S} |\textbf{i}\rangle .
\end{eqnarray}
Since for $\textbf{i} \in J_k$,
\begin{equation}
\langle \textbf{i} |\hat{S} |\textbf{i}\rangle = k
\end{equation}
we find
\begin{equation}
\delta = \sum_{k=0}^N \sum_{\textbf{i}\in J_k} |c_{\textbf{i}}|^2\, k \geq \sum_{k=1}^N \sum_{\textbf{i}\in J_k} |c_{\textbf{i}}|^2
\end{equation}
and alternatively also
\begin{equation}
\delta \leq N\, \sum_{k=1}^N \sum_{\textbf{i}\in J_k} |c_{\textbf{i}}|^2 .
\end{equation}
The normalization of $|\Psi\rangle$, $\sum_{\textbf{i}} |c_{\textbf{i}}|^2 =1$ yields ($J_0 =\{(1,2,\ldots,N)\}$ contains only one element)
\begin{equation}
\frac{\delta}{N} \leq 1-|c_{(1,\ldots,N)}|^2 \leq \delta .
\end{equation}
and thus
\begin{equation}
1-\delta \leq |c_{(1,\ldots,N)}|^2 \leq 1-\frac{1}{N}\delta .
\end{equation}
\end{proof}

Now, we come back to the original question.
We first state the mathematical result and present the proof afterwards.
\begin{theorem}\label{HFworks}
Given a state $|\Psi\rangle \in \wedge^3[\mathcal{H}_6]$ with natural occupation numbers $(\lambda_k)_{k=1}^6$.
Let $P$ be the projection operator onto the subspace spanned by the states $|1,2,3\rangle, |1,4,5\rangle, |2,4,6\rangle$, which corresponds to exact pinning
of $D^{(6)}=\lambda_5+\lambda_6-\lambda_4 \geq 0$ (recall (\ref{HFextstate})). Then as long as
\begin{equation}
\delta\equiv 3-\lambda_1-\lambda_2-\lambda_3 \leq \frac{1}{4}
\end{equation}
(which means nothing else but being not too far away from the spectrum $\lambda_{\mathrm{Sl}}=(1,1,1,0,0,0)$ of a single Slater determinant) we find
\begin{equation}
1- \chi_{\delta} D^{(6)}\leq \|P \Psi\|_2^2 \leq 1- \frac{1}{2} D^{(6)} ,
\end{equation}
with
\begin{equation}
\chi_{\delta} \equiv \frac{1+2\delta}{1 - 4\delta} .
\end{equation}
\end{theorem}

\begin{proof}
In Section \ref{sec:selectionrule} we concluded that $|\Psi\rangle$ has the form
\begin{eqnarray}\label{ansatz}
|\Psi\rangle &=& \alpha |1,2,3\rangle+ \beta |1,2,4\rangle+ \gamma |1,3,5\rangle \nonumber \\
 &&+ \,\delta |2,3,6\rangle +\nu |1,4,5\rangle +\mu |2,4,6\rangle  \nonumber \\
 &&+ \,\xi |3,5,6\rangle+\zeta |4,5,6\rangle \,,
\end{eqnarray}
with natural orbitals $\{|k\rangle\}_{k=1}^6$.
Since the corresponding $1-$RDO is diagonal w.r.t. $\{|k\rangle\}_{k=1}^6$,
\begin{equation}\label{diagonal}
\langle k |\rho_1|l\rangle = \delta_{k l} \,\lambda_k ,
\end{equation}
we find (recall (\ref{noncoef}))
\begin{eqnarray}
\lambda_4 &=& |\beta|^2+|\nu|^2+|\mu|^2+|\zeta|^2 \\
\lambda_5 &=& |\gamma|^2+|\nu|^2+|\xi|^2+|\zeta|^2 \\
\lambda_6 &=& |\delta|^2+|\mu|^2+|\xi|^2+|\zeta|^2
\end{eqnarray}
The goal is now to show that the coefficients $\beta, \gamma, \delta,\xi$ and $\zeta$ are small, i.e.
\begin{eqnarray}
\|P \Psi\|_{L^2}^2 &=& |\alpha|^2+ |\mu|^2+|\nu|^2 \nonumber \\
&=& 1-\left( |\beta|^2+ |\gamma|^2+|\delta|^2+|\xi|^2+ |\zeta|^2\right)
\end{eqnarray}
is close to $1$, whenever (\ref{d=6d}), which here reads
\begin{equation}\label{distancegreek}
D^{(6)} = -|\beta|^2+|\gamma|^2+|\delta|^2+ 2|\xi|^2+|\zeta|^2 , \\
\end{equation}
is approximately saturated.
First we observe
\begin{eqnarray}\label{overlapgreek}
\|P \Psi\|_{L^2}^2 &\leq& 1- \frac{1}{2}\left( |\beta|^2+ |\gamma|^2+|\delta|^2+ 2|\xi|^2+ |\zeta|^2\right)\nonumber\\
&\leq& 1- \frac{1}{2}\left( - |\beta|^2+ |\gamma|^2+|\delta|^2+ 2|\xi|^2+ |\zeta|^2\right)\nonumber\\
&=& 1-\frac{1}{2}\,D^{(6)} ,
\end{eqnarray}
which is the upper bound for $\|P \Psi\|_{L^2}^2$ in Theorem \ref{HFworks}.

To derive the lower bound note the essential difference in (\ref{distancegreek}) and (\ref{overlapgreek}), the sign of the term $|\beta|^2$.
To get rid of this we write $|\beta|^2 = - \chi\, |\beta|^2 + (1+ \chi) |\beta|^2$, $\chi>0$ and estimate $(1+\chi)|\beta|^2$ in terms of $|\gamma|^2, |\delta|^2, |\xi|^2, |\zeta|^2$.
For this observe that (\ref{diagonal}) in particular implies
\begin{eqnarray}
0 = \langle 4|\rho_1|3\rangle = \overline{\alpha} \beta +  \overline{\gamma} \nu +  \overline{\delta} \mu +  \overline{\xi} \zeta\,,
\end{eqnarray}
which leads by the triangle inequality, the identity $(A+B+C)^2 \leq 3 \,(A^2+B^2 + C^2)$ and $|\mu|^2, |\nu|^2,|\xi|^2,|\zeta|^2 \leq 1-|\alpha|^2$ to
\begin{eqnarray}\label{betaestimate}
|\beta|^2 &=& \left|\frac{1}{\overline{\alpha}}\,(\overline{\gamma} \nu +  \overline{\delta} \mu +  \overline{\xi} \zeta)\right|^2 \nonumber \\
&\leq& \frac{1}{|\alpha|^2}\,\left(|\gamma| \,|\nu| + |\delta| \,|\mu| +|\xi| \,|\zeta|   \right)^2                  \nonumber \\
&\leq& \frac{3}{|\alpha|^2}\,\left(|\gamma|^2 \,|\nu|^2 + |\delta|^2 \,|\mu|^2 +|\xi|^2 \,|\zeta|^2   \right)\nonumber \\
 &\leq& \frac{3(1-|\alpha|^2)}{|\alpha|^2}\,\left(|\gamma|^2  + |\delta|^2  + \frac{1}{3} (2|\xi|^2 +|\zeta|^2) \right)\,.
\end{eqnarray}
Now, for all $s,r \geq0$ we find by using (\ref{betaestimate})
\begin{eqnarray}\label{lowerboundest}
\lefteqn{|\beta|^2+ |\gamma|^2+|\delta|^2+|\xi|^2+ |\zeta|^2}&&\nonumber \\
&\leq& (1-r)|\beta|^2+ |\gamma|^2+|\delta|^2+ (1+s)(2 |\xi|^2+ |\zeta|^2) + r |\beta|^2 \nonumber \\
&\leq& (1-r)|\beta|^2+ |\gamma|^2+|\delta|^2+(1+s)(2 |\xi|^2+ |\zeta|^2) \nonumber \\
&& + \frac{3 r (1-|\alpha|^2)}{|\alpha|^2}\,\left(|\gamma|^2  + |\delta|^2  + \frac{1}{3} (2|\xi|^2 +|\zeta|^2) \right) \nonumber \\
&=& (1-r)|\beta|^2 + \left(1+ \frac{3r(1-|\alpha|^2)}{|\alpha|^2}\right)\,\left(|\gamma|^2 + |\delta|^2\right) \nonumber \\
&& +  \left( 1+ s+\frac{r(1-|\alpha|^2)}{|\alpha|^2}\right) \left(2 |\xi|^2+ |\zeta|^2\right) .
\end{eqnarray}
By choosing
\begin{eqnarray}
r&=& \frac{2 |\alpha|^2}{4 |\alpha|^2-3}\\
s&=& \frac{4 (1-|\alpha|^2)}{4|\alpha|^2-3}
\end{eqnarray}
the last expression in (\ref{lowerboundest}) coincides with $D^{(6)}$ up to a global factor $\chi$.
Both parameters $r,s$ are non-negative as long as $|\alpha|^2\geq \frac{3}{4}$.
Finally, this leads to
\begin{eqnarray}
\|P \Psi\|_{L^2}^2 &=& 1-(|\beta|^2+ |\gamma|^2+|\delta|^2+|\xi|^2+ |\zeta|^2) \nonumber \\
&\geq& 1-(r-1)D^{(6)} \nonumber \\
&\equiv& 1-\chi_{1-|\alpha|^2} D^{(6)} ,
\end{eqnarray}
with
\begin{eqnarray}
\chi_{1-|\alpha|^2}\equiv r-1 &=& \frac{3-2 |\alpha|^2}{4 |\alpha|^2-3} \nonumber \\
&=&  \frac{1 + 2(1- |\alpha|^2)}{1-4(1-|\alpha|^2)} .
\end{eqnarray}
Lemma \ref{Slaterstable} states $|\alpha|^2 \geq 1-\delta$ and since $\chi$ is monotonously increasing, $\chi_{1-|\alpha|^2}\leq \chi_{\delta}$,
which finishes the proof.
\end{proof}

\bibliography{refqmpArXiv,comments}

\end{document}